\documentclass[aps,preprintnumbers,nofootinbib,superscriptaddress,twocolumn]{revtex4}

\usepackage{amsmath}
\usepackage{amsthm}
\usepackage{amssymb}
\usepackage{amstext}
\usepackage{graphicx}
\usepackage{url}
\usepackage{multirow}

\theoremstyle{plain}
\newtheorem{theorem}{Theorem}
\newtheorem{lemma}[theorem]{Lemma}

\theoremstyle{definition}
\newtheorem{definition}[theorem]{Definition}

\newcommand*{\comment}[1]{}

\newcommand*{\ket}[1]{| #1 \rangle}

\newcommand*{\cG}{\mathcal{G}}

\newcommand*{\ID}{{\sc id}} 
\newcommand*{\RE}{{\sc r}} 
\newcommand*{\REt}{{\sc s}} 
\newcommand*{\REc}{{\sc t}} 

\newcommand*{\IDp}{{\sc id}$'$}
\newcommand*{\REp}{{\sc r}$'$}
\newcommand*{\REpl}{{\sc r}${}^+$}

\newcommand*{\QM}{\text{{\sc qm}}}

\begin{document}

\title{Unconditionally secure device-independent quantum key
  distribution\\ with only two devices}

\author{Jonathan \surname{Barrett}}
\email[]{jon.barrett@rhul.ac.uk}
\affiliation{Royal Holloway, University of London, Egham Hill, Egham TW20 0EX, U.K.}
\author{Roger \surname{Colbeck}}
\email[]{colbeck@phys.ethz.ch}
\affiliation{Institute for Theoretical Physics, ETH Zurich, 8093
Zurich, Switzerland.}
\author{Adrian \surname{Kent}}
\email[]{a.p.a.kent@damtp.cam.ac.uk}
\affiliation{Centre for Quantum Information and Foundations, DAMTP, Centre for Mathematical Sciences, University of Cambridge, Wilberforce Road, Cambridge, CB3 0WA, U.K.}
\affiliation{Perimeter Institute for Theoretical Physics, 31 Caroline Street North, Waterloo, ON N2L 2Y5, Canada.}
\date{$11^{\text{th}}$ October 2012}

\begin{abstract}
  Device-independent quantum key distribution is the task of using
  uncharacterized quantum devices to establish a shared key between
  two users.  If a protocol is secure regardless of the device
  behaviour, it can be used to generate a shared key even if the
  supplier of the devices is malicious.  To date, all
  device-independent quantum key distribution protocols that are known
  to be secure require separate isolated devices for each entangled
  pair, which is a significant practical limitation.  We introduce a
  protocol that requires Alice and Bob to have only one device
  each. Although inefficient, our protocol is unconditionally secure
  against an adversarial supplier limited only by locally enforced
  signalling constraints.
\end{abstract}

\maketitle

\section*{Introduction}
Key distribution is the task of establishing shared secret strings
between two parties, and is sufficient for secure communication.
Classical key distribution protocols base their security on
assumptions about an eavesdropper's computational power.  On the other
hand, quantum key distribution protocols (e.g.~\cite{BB84,Ekert})
promise security against an arbitrarily powerful eavesdropper, and do
so in the presence of realistic noise levels.  However, in order for
the security proofs to apply, the devices must operate according to
certain specifications.  Deviations from these can introduce security
flaws, which can be difficult to identify (see e.g.~\cite{GLLSKM} for
practical illustrations of such attacks).

The difficulty associated with verifying the operation of quantum
devices has led to much interest in device-independent quantum
cryptography protocols.  Ideally, such protocols guarantee security by
tests on the outputs of the devices: no specification of their
internal functionality is required.  In a sense, the protocol verifies
the devices' security \emph{on-the-fly}.

Device-independent cryptography was first introduced by Mayers and
Yao~\cite{MayersYao} (albeit under a different name) and the first
quantum key distribution protocol to be proven device-independently
secure was the Barrett-Hardy-Kent (BHK) protocol~\cite{BHK}.  The BHK
security proof applies not only against an arbitrarily powerful
quantum eavesdropper (who also supplies the devices) but even against
an eavesdropper and device-supplier who has discovered and makes use
of any post-quantum physical theory, provided that, within the theory,
the honest parties can enforce local signalling constraints.  The
applicability of the BHK protocol and proof to device-independent
quantum cryptography was explicitly pointed out by later authors, who
went on to develop some more efficient device-independent protocols
with security proofs against restricted
eavesdroppers~\cite{agm,SGBMPA,abgmps} as well as other protocols
shown to be unconditionally secure~\cite{MRC,Masanes,HRW2,HR,MPA}.

From a theoretical perspective, the BHK protocol provided an existence
theorem for a task that had not been known to be possible.
Practically, however, it has drawbacks.  One is that, as formulated,
it generates only a single bit of secure key.  Although it can be
modified using an idea from~\cite{BKP} to produce an arbitrarily long
key, even with this modification, the protocol is inefficient and
unable to tolerate reasonable levels of noise.

A serious practical problem with all the protocols with proven
unconditional device-independent security~\cite{BHK,MRC,HR,MPA} is
that they require that each (purportedly) entangled pair used in the
protocol is isolated from the others. The protocols thus require a
separate and isolated pair of devices for each entangled pair to
ensure full device-independent security.  This evidently makes such
protocols costly to implement in practice.

We introduce here a protocol that evades this limitation, requiring
only a single device for each user.  Our protocol is a refinement of
the BHK protocol, necessary in order to allow security when used with
only two devices.  As we have discussed elsewhere~\cite{bckone}, the
composability of device independent protocols is problematic if
devices are reused in subsequent implementations. Here we show that if
devices are not reused, then our protocol is secure according to a
universally composable security definition, even against an adversary
who supplies the devices and is restricted only by signalling
constraints. As described, our protocol generates a single secure key
bit.  We also indicate how it can be modified using the idea
in~\cite{BKP}, to produce a key of arbitrary length. In addition,
since it is composable, further key bits can be generated by running
the protocol several times (although in this case, fresh devices are
required for each run). 

We see the value of our protocol as an existence theorem showing that
device-independent quantum key distribution {\it is} in principle
possible with only two devices.  Whether this task can be achieved
more efficiently and with reasonable noise tolerance remains (as far
as we are aware) an open question.

We also show that some apparently natural extensions of existing
protocols to two devices are insecure against eavesdroppers restricted
only by signalling constraints, and in some cases also against quantum
eavesdroppers.  This may have impact in a recent line of work on the
impossibility of privacy amplification against non-signalling
eavesdroppers~\cite{HRW,AHT}.

\section*{Cryptographic scenario}
We use a standard cryptographic scenario for key distribution.  Here,
two users (Alice and Bob), each have a secure laboratory in which to
work, which they may partition into secure sub-laboratories. These
allow Alice and Bob to prevent unauthorized communications between any
devices they use.  They are also each assumed to have (or be able to
generate) their own supply of trusted random bits.  To communicate
between one another, Alice and Bob have access to an authenticated,
but insecure, classical channel, and an insecure quantum channel.
They may process classical information in a trusted way within their
laboratories.  However, any devices they use for quantum information
processing are assumed to be supplied by an untrusted adversary (Eve).
Eve may access (but not modify) any classical correspondence between
Alice and Bob, and may access and modify quantum communication between
them.  She has complete knowledge of the protocol, but does not have
access to the classical random data that Alice and Bob generate within
their labs and use for the protocol (except for information she can
deduce from what they make public).

\section*{Setup for the protocol}
Alice and Bob each have a device, potentially supplied by Eve, that
has an input port with $N\geq 2$ possible inputs and an output port
with $2$ possible outputs.  Alice's inputs are denoted
$A\in\{0,2,\ldots, 2N-2\}$, and Bob's $B\in\{1,3,\ldots, 2N-1\}$, and
their respective outputs are denoted $X\in\{0,1\}$ and $Y\in\{0,1\}$.
We define a set of allowed input pairs $(A,B)$ by $\cG_N:=\{(0,2N-1),
(0,1), (2,1), (2,3), \ldots ,(2N-2,2N-1) \}$, with $|\cG_N|=2N$.  For
convenience, we introduce $X'$ as a variable that is equal to $1-X$ if
$(A,B)=(0,2N-1)$, and equal to $X$ otherwise.

The devices are claimed by Eve to function by carrying out specified
binary outcome measurements on the maximally entangled two qubit state
$\ket{\Phi^+} = \frac{1}{\sqrt{2}}(\ket{00}+\ket{11})$.
Alice's input $A$ is claimed to correspond to measuring the first
qubit in the basis
$\{\cos\frac{\theta}{2}\ket{0}+\sin\frac{\theta}{2}\ket{1},\
\sin\frac{\theta}{2}\ket{0}-\cos\frac{\theta}{2}\ket{1}\}$, where
$\theta=\frac{\pi A}{2N}$; similarly Bob's input $B$ is claimed to
correspond to measuring the second qubit in the basis defined by
$\theta=\frac{\pi B}{2N}$.  Alice and Bob do not need to test these
precise claims, but instead perform various measurements and check
their outcomes in such a way that the checks are unlikely to pass
unless the produced bit is virtually as secure as a bit that would be
generated were Eve's claims correct.

The protocol involves two security parameters: the integer $N\geq 2$
defined above, and a real number $\alpha$ in the range $0<\alpha<1$:
to achieve reasonable security $N$ needs to be large and $\alpha$
small.  All classical communication between Alice and Bob is done via
their authenticated classical channel.

Throughout the protocol, Alice and Bob keep their devices in isolated
parts of their secure laboratories, ensuring that each device only
learns its own inputs and cannot send any information outside the
secure area.  This ensures that the behaviour of the devices, which
can be specified by a conditional probability distribution, satisfies
certain non-signalling constraints.  In particular, if the system
Alice and Bob measure is correlated with a third system with input $C$
and outcome $Z$, then the overall behaviour of the devices,
$P_{XYZ|ABC}$, must be non-signalling, i.e.\ satisfy
\begin{eqnarray}\label{tripartitenonsig}
P_{XY|ABC} &=& P_{XY|AB}\\
P_{YZ|ABC} &=& P_{YZ|BC} \nonumber \\
P_{XZ|ABC} &=& P_{XZ|AC} \nonumber \, .
\end{eqnarray}
These conditions ensure that if three parties possess devices with
this behaviour, no subset of the parties can signal to any other
subset by varying their choice of input.

\subsection*{Protocol~\RE}
\begin{enumerate}
\item \label{step1} Alice randomly chooses $K$, such that $K=0$ with
  probability $1-\alpha$ and $K=1$ with probability $\alpha$.  She
  announces $K$ to Bob.
\item On the $i^{\text{th}}$ round, Alice picks a pair of values
  $(A_i,B_i)$ at random from the set $\cG_N$ specified above, and
  announces them both to Bob\footnote{In fact, Alice need only
    announce $B_i$, but we have her announce both to make the analysis
    simpler.}.
\item Alice inputs $A_i$ into her device, and Bob $B_i$ into his, and
  they record their outputs, the bits $X_i$ and $Y_i$
  respectively. (Alice ensures that her device doesn't learn $B_i$.)
  If $(A_i,B_i)=(0,2N-1)$, Alice sets $X'_i=1-X_i$, otherwise
  she sets $X'_i=X_i$.
\item \label{step4} If $K=0$, Alice and Bob announce
  $X'_i$ and $Y_i$.  If $X'_i\neq Y_i$, they abort.
  Otherwise, they return to Step~\ref{step1}.
\item \label{step6} If $K=1$, write $i=f$ (the final value of $i$).
  The bits $X'_f$ and $Y_f$ are taken to be the final shared
  secret key bit.
\end{enumerate}

As presented above, this protocol requires Alice's and Bob's devices
to contain sufficient pre-shared entanglement before the protocol
starts.  Taken literally, this requires an infinite supply of
pre-shared $\ket{\Phi^+}$ states.  More realistically, it requires a
large number $M \gg \alpha^{-1}$ of pre-shared $\ket{\Phi^+}$ states,
and that the parties accept a small probability of the protocol
aborting because the supply is exhausted.  These stringent
technological requirements can be avoided by introducing an additional
(untrusted) state-creation device, which could be incorporated into
Alice's or Bob's measurement device, and which is supposed to generate
$\ket{\Phi^+}$ states and send one qubit over the insecure quantum
channel to the other party.  The $i^{\text{th}}$ state must be
distributed before any information about the measurements $(A_i, B_i)$
or the value of $K$ is announced.  This modification (call it
Protocol~\REpl ) gives Eve more cheating strategies but, as we show below, 
is still secure.  

\section*{Security -- main idea}
The idea behind the security of this protocol is as follows.  If the
states and measurements are as Eve claims, then the quantity $I_N$
defined by
\begin{align*}
I_N=I_N(P_{XY|AB}):=P(X=Y|A=0,B=2N-1)+\\\sum_{\genfrac{}{}{0pt}{}{a,b}{|a-b|=1}}P(X\neq
Y|A=a,B=b)
\end{align*}
satisfies 
\begin{equation} \label{qcbi}
I_N=I_N^{\QM}:=2N\sin^2\frac{\pi}{4N}<\frac{\pi^2}{8N} \, .
\end{equation}
As $N$ increases, these correlations give larger violations the
chained Bell inequalities~\cite{Pearle,BC}, which in this formulation
are $I_N\geq 1$.

The significance of this violation of the chained Bell inequalities
for secrecy is that, in the limit of large $N$, the correlations that
achieve the quantum bound~\eqref{qcbi} become monogamous and
uniform~\cite{BHK,BKP}.  That is, for any non-signalling distribution
$P_{XYZ|ABC}$ for which $I_N(P_{XY|AB})$ is small, and for any choice
of input $c$, the outcome $Z$ is virtually uncorrelated with $X$, and
$P_{X|A=a}$ is virtually indistinguishable from uniform, for all
$a$.\footnote{A note about that notation used in this paper.  We tend
  to use upper case for random variables, and lower case for
  particular instances of them.  In addition, $P_{X|A=a}$ is the
  distribution over the random variable $X$ conditioned on the event
  that random variable $A$ takes value $a$.  This will often be
  abbreviated to $P_{X|a}$.  There is another common notation in which
  this is written $P(X|A=a)$.} In other words, if Alice's and Bob's
systems have a low $I_N$, then Eve (who we can take to hold the system
with input $C$ and output $Z$) must have almost no information about
the outcomes they obtain.  The protocol is designed so that (roughly
speaking) if Eve supplies states for which there are many rounds in
the protocol where $I_N$ is high, the protocol is likely to abort,
while if she supplies a state that has high $I_N$ on only a few
rounds, the round at which Alice and Bob finally (hope to) create the
key bit is likely to have low $I_N$, and so the key bit is likely to
be indeed both agreed by Alice and Bob and secure against Eve.

Our main result is that, if we choose $\alpha=N^{-\frac{3}{2}}$, and
take $N$ to be large, Protocol~{\RE} is unconditionally secure, in the
sense that the key bit it generates can be treated as though produced
by a secure random key distribution oracle.  Provided that the devices
are not reused and are securely isolated, so that secret information
generated in the protocol cannot subsequently be made
public~\cite{bckone}, this also shows that the generated key bit is
composably secure.

Although the protocol generates only a single key bit, it can be
simply modified to generate more key bits, still using only two
devices, based on correlations introduced in~\cite{BKP}.  The modified
protocol uses devices with $L>2$ outcomes on each side and $I_N$ is
replaced by the quantity
\begin{eqnarray*}
I_{N,L}(P_{XY|AB})&:=&P(X\oplus_L 1\neq Y|A=0,B=2N-1)+\\
&&\sum_{\genfrac{}{}{0pt}{}{a,b}{|a-b|=1}}P(X\neq Y|A=a,B=b),
\end{eqnarray*}
where $X\oplus_L 1$ represents addition modulo $L$.  This protocol can
be implemented by quantum devices containing maximally entangled
$L$-dimensional quantum states and carrying out measurements with $L$
possible outcomes~\cite{BKP}.

The next section contains a precise statement and proof of security
for Protocol~{\RE}.  

\section*{Security definition}
We use here a standard definition of composable security (based on the
definitions in~\cite{Can01}, previously applied in an analogous way to
our treatment in~\cite{HRW2,HR,Hanggi}).  A composable security
definition should ensure that a protocol is not only secure for a
single instance, but that it remains secure if used as a sub-protocol
in part of an arbitrary extended protocol.  In order to show this, one
considers an ideal protocol (that is by definition secure) and proves
that there is no extended protocol that can correctly guess whether it
is interfacing with the ideal or real protocol with probability
significantly greater than~$\frac{1}{2}$.  Roughly speaking, the idea
is that if this holds, the two protocols behave essentially
identically when used as part of any other protocol.  Furthermore, if
the probability of correctly guessing differs from~$\frac{1}{2}$ by at
most $p$, then, for $n$ uses of either the real protocol or idea, the
probability of correctly guessing differs from~$\frac{1}{2}$ by at
most $np$.

Formally, one considers a \emph{distinguisher}, that tries to guess
which protocol (the real or ideal) is being used.  For two key
distribution protocols, 1 and 2, a distinguisher is an extended
protocol that uses the candidate protocol as a sub-protocol, and
outputs a single bit, corresponding to a guess of whether the
sub-protocol was protocol 1 or 2.  The distinguisher can ask the
eavesdropper to act in any way\footnote{Note that what Eve does can be
  adapted depending on any information available to the
  distinguisher.}.  and can use Eve's outputs, those of the honest
parties and any information made public in the protocol's
implementation to try to distinguish the two.  It does not, however,
have access to any private data that the honest users use.

Let us denote by $\Gamma$ the complete set of random variables the
distinguisher receives from Alice and Bob during the protocol, as well
as the protocol's outputs.  If Protocol~1 is followed, these are
distributed according to $Q^1_{\Gamma}$, while if Protocol~2 is
followed, these are distributed according to $Q^2_{\Gamma}$ (for some
fixed device behaviour chosen by Eve).  Having received these, the
distinguisher has access to a system (held by Eve) with input denoted
$C$, and output $Z$.  The probability of correctly guessing whether
Alice and Bob are following Protocol~1 or~2 (chosen with probability
$\frac{1}{2}$ each) is given by\footnote{Note that in the case that
  Eve keeps only a classical system (so there is no $c$), this reduces
  to $\frac{1}{2}(1+D(P^1_{\Gamma Z},P^2_{\Gamma Z}))$, where $D$
  denotes the \emph{total variation distance} (defined later).}.
\begin{equation}\label{eq:defn}
\frac{1}{2}(1+\frac{1}{2}\sum_{\gamma}\max_c\sum_z|Q^1_{\Gamma}(\gamma)Q^1_{Z|\gamma
  c}(z)-Q^2_{\Gamma}(\gamma)Q^2_{Z|\gamma c}(z)|)\, .
\end{equation}
The notion of security we use is based on the success probability of
the optimal distinguisher (i.e.\ where the distinguisher asks Eve to
behave in such a way as to make distinguishing easiest)\footnote{A
  note on notation: we characterize the behaviour of the devices by
  the joint conditional probabilities of the outputs if the inputs are
  chosen independently, and label these using $P$.  For example, in
  the case of three devices shared between Alice, Bob and Eve, these
  are denoted $P_{XYZ|ABC}$ and are assumed to satisfy the
  no-signalling conditions (\ref{tripartitenonsig}).  
 We use expressions involving $Q$ (e.g.,
  $Q_{\Gamma CZ}$) to denote the actual distribution of random
  variables in the scenario where a protocol is being performed on
  these systems in conjunction with a distinguisher.  There is an
  important distinction between the two: since a distinguisher can
  arrange that $C$ is correlated with $\Gamma$, $Q$ may no longer obey
  the no-signalling conditions  (\ref{tripartitenonsig}). 
   For example, if $\Gamma$ includes the
  output, $X$, of Alice's device (whose input is $A$), and the
  distinguisher chooses $C=X$, the non-signalling condition
  $Q_{X|AC}=Q_{X|A}$ does not generally hold.}.

\begin{definition}\label{def:zeta} 
  Protocol 1 is said to be \emph{$\zeta$-secure} with respect to
  Protocol 2 if the probability of correctly guessing whether a
  candidate protocol is Protocol 1 or 2 (chosen with probability
  $\frac{1}{2}$ each) by any distinguisher is at most
  $\frac{1}{2}(1+\zeta)$.
\end{definition}

We define an ideal protocol, Protocol~{\ID}, to be identical to
Protocol~{\RE}, except that Step~\ref{step6} is replaced by
\begin{enumerate}
\item[\ref{step6}$'$.] If $K=1$, Alice and Bob take their outputs from a
  hypothetical device that gives $X$ to Alice, and $Y$ to Bob such
  that $X=Y$ and $X$ is uniformly distributed and uncorrelated with
  any other information.
\end{enumerate}
This protocol either aborts (with the same probability as
Protocol~{\RE}), or outputs the same perfectly private bit to both
Alice and Bob.

In order to prove security of Protocol~{\RE}, it is useful to define a
modified protocol, to be used as a technical tool in the proof.  We
consider a protocol that is the same as Protocol~{\RE}, except with a
more powerful eavesdropper who, before the protocol restarts at the
end of Step~\ref{step4}, has access to all the data previously
produced and can alter Alice's and Bob's devices at this stage.
Formally, let Protocol~{\REp} be identical to Protocol~{\RE}, except
that Step~\ref{step4} is replaced by
\begin{enumerate}
\item[\ref{step4}$'$.] If $K=0$, Alice and Bob publicly announce their
  outputs $X'_i$ and $Y_i$.  If $X'_i\neq Y_i$, they abort.
  Otherwise, they return their devices to Eve who can modify them and
  supply new ones.  Alice and Bob both announce receipt of their new
  devices, before returning to Step~\ref{step1}.
\end{enumerate}
We also define an analogous ideal, Protocol~{\IDp}, which is obtained
from Protocol~{\ID} by replacing Step~\ref{step4} with
Step~\ref{step4}$'$.

The reason for this adjustment is that Protocol~{\REp} clearly cannot
be more secure than Protocol~{\RE} (the set of allowed actions of Eve
in Protocol~{\REp} is strictly larger than that in
Protocol~{\RE}). Hence it is sufficient to prove security of
Protocol~{\REp}. But the analysis of Protocol~{\REp} is relatively
simple, because the optimal distinguisher will ask the eavesdropper to
act in an independent and identically distributed (i.i.d.) way on each
round, and is essentially characterized by the single constant value
of $I_N$ used on each round.

We will show that Protocol~{\REp} is $\zeta$-secure with respect to
Protocol{~\IDp}, where the parameter $\zeta$ can be made arbitrarily
small by appropriate choices of $\alpha$ and $N$.  Since both
protocols have identical probabilities of aborting, an abort event
cannot help the distinguisher.  Furthermore, in any strategy with a
significant probability of not aborting, the protocols remain
virtually indistinguishable.  This shows that Protocol~{\REp} is
composably secure in the appropriate sense.  As mentioned before, it
follows that Protocol~{\RE} is also $\zeta$-secure with respect to
Protocol{~\ID} and hence also composably secure.

\section*{Security proof}
The proof bounds the probability of distinguishing Protocols~{\REp}
and~{\IDp}.  First, note that there is an optimal distinguishing
strategy in which Eve's actions are i.i.d.\ since, if it does not
abort, when the protocol returns to Step~\ref{step1}, the maximum
probability of distinguishing the protocols is identical to that
before the protocol began.

We use the following lemma, that uses $I_N$ to bound the distance
between probability distributions, measured using the total variation
distance, $D(P_X,Q_X):=\frac{1}{2}\sum_x|P_X(x)-Q_X(x)|$.  The proof
of this lemma can be found in [\citenum{CR_ext}, Supplementary
Information] (and is based on similar results
in~\cite{BHK,BKP,ColbeckRenner}):
\begin{lemma}\label{lem:1} \cite{CR_ext}
  For any non-signalling device behaviour, $P_{XYZ|ABC}$, in
  which $X$ and $Y$ are binary, we have
\begin{equation}\label{eq:f}
D(P_{Z|abcx},P_{Z|c})\leq I_N(P_{XY|AB})
\end{equation}
for all $a$, $b$, $c$ and $x$, and
\begin{equation}
D(P_{X|abc},\frac{1}{2})\leq\frac{1}{2}I_N(P_{XY|AB})
\end{equation}
for all $a$, $b$ and $c$.
\end{lemma}
(Note: We use $D(P_{X|abc},\frac{1}{2})$ to denote the distance
between $P_{X|abc}$ and the distribution where $X=0$ and $X=1$ both
occur with probability $\frac{1}{2}$.)


Note that these relations imply
\begin{equation}\label{eq:DIN}
D(P_{Z|abc,X'=x},P_{Z|c})\leq I_N(P_{XY|AB})
\end{equation}
and
\begin{equation}\label{eq:distX}
D(P_{X'|ab},\frac{1}{2})\leq\frac{1}{2}I_N(P_{XY|AB})\, .
\end{equation}

Note also that, from the definition of $I_N$, averaging over the
measurements in $\cG_N$ (picked uniformly), we have
\begin{eqnarray}
P(X'\neq
Y)&:=&\!\!\!\sum_{\genfrac{}{}{0pt}{}{abx}{(a,b)\in\cG_N}}\!\!\!\frac{P_{X'Y|ab}(x,1\!-\!x)}{2N}\nonumber\\
&=&\frac{I_N(P_{XY|AB})}{2N}\, .\label{correctbound}
\end{eqnarray}

We also need the following generalization of~\eqref{eq:DIN}:
\begin{lemma}
  For any non-signalling device behaviour, $P_{XYZ|ABC}$, in
  which $X$ and $Y$ are binary, and $I_N:=I_N(P_{XY|AB})<1$ we have that,
  for $(a,b)\in\cG_N$,
\begin{eqnarray}\label{eq:withY}
D(P_{Z|abc,X'=x,Y=x},P_{Z|abc,X'=x})\leq\frac{2I_N}{1-I_N}\, .
\end{eqnarray}
\end{lemma}
\begin{proof}
We have
\begin{widetext}
\begin{align*}
P_{Z|abc,X'=x,Y=x}&(z)-P_{Z|abc,X'=x}(z)=P_{Z|abc,X'=x,Y=x}(z)-\sum_yP_{YZ|abc,X'=x}(y,z)\\
&=P_{Z|abc,X'=x,Y=x}(z)-P_{Y|abc,X'=x}(x)P_{Z|abc,X'=x,Y=x}(z)-P_{Y|abc,X'=x}(1-x)P_{Z|abc,X'=x,Y=1-x}(z)\\
&=(1-P_{Y|abc,X'=x}(x))(P_{Z|abc,X'=x,Y=x}(z)-P_{Z|abc,X'=x,Y=1-x}(z))
\end{align*}
and hence
\begin{eqnarray*}
D(P_{Z|abc,X'=x,Y=x},P_{Z|abc,X'=x})&=&(1-P_{Y|abc,X'=x}(x))D(P_{Z|abc,X'=x,Y=x}(z),P_{Z|abc,X'=x,Y=1-x}(z))\\
&\leq&(1-P_{Y|abc,X'=x}(x))\, .
\end{eqnarray*}
\end{widetext}
Then note that averaging over the measurements in $\cG_N$,
using~\eqref{correctbound} we have
\begin{align*}
1-\frac{I_N}{2N}&=\sum_{\genfrac{}{}{0pt}{}{a'b'x}{(a',b')\in\cG_N}}\frac{1}{2N}P_{X'Y|a'b'}(x,x)\\
&\leq\frac{1}{2N}\left(\sum_xP_{X'Y|ab}(x,x)+2N-1\right)\, ,
\end{align*}
from which it follows that
\begin{align*}
\sum_xP_{X'Y|ab}(x,x)\geq 1-I_N\, ,
\end{align*}
and hence
\begin{align*}
P&_{Y|abc,X'=x}(x)=\frac{P_{X'Y|abc}(x,x)}{P_{X'|abc}(x)}\\
&\geq\frac{1}{P_{X'|ab}(x)}\left(1-I_N-P_{X'Y|abc}(1-x,1-x)\right)\\
&\geq\frac{1}{P_{X'|ab}(x)}\left(1-I_N-(1-P_{X'|ab}(x))\right)\\
&\geq1-\frac{2I_N}{1-I_N}\, ,
\end{align*}
where we used~\eqref{eq:distX} in the last line.  Note that the last
step does not hold unless $I_N<1$.  The claimed relation then follows.
\end{proof}
Combining~\eqref{eq:withY} and~\eqref{eq:DIN} (using the triangle
inequality for $D$), we have for $I_N<1$
\begin{equation}\label{eq:together}
D(P_{Z|abc,X'=x,Y=x},P_{Z|c})\leq (1+\frac{2}{1-I_N})I_N\, .
\end{equation}

To successfully distinguish the protocols it is necessary that they do
not abort before the final round.  We use $\bot$ to represent the
event that the protocol aborts, and $\bar{\bot}$ to represent the
event that it does not.

\begin{lemma}\label{lem:2}
  For $0\leq I_N^*\leq 2N$, if Protocol~{\REp} is followed, and Eve
  supplies i.i.d.\ states corresponding to non-signalling
  device behaviours with $I_N(P_{X_iY_i|A_iB_i}) = I_N^*$ for all $i$,
  then
$$Q(\bar{\bot})=  \left(1+\frac{(1-\alpha)I_N^*}{2N\alpha}\right)^{-1}.$$
\end{lemma}
\begin{proof}
  We have
$$Q(f=j)=\left((1-\alpha)(1-\frac{I_N^*}{2N})\right)^{j-1}\alpha \, ,$$
and hence
\begin{eqnarray*}
Q(\bar{\bot} )
= \sum_{j=1}^{\infty}Q(f=j)=\left(1+\frac{(1-\alpha)I_N^*}{2N\alpha}\right)^{-1}\, , 
\end{eqnarray*}
as required.
\end{proof}

Our main result is then as follows
\begin{theorem}
  Take $\alpha=N^{-\frac{3}{2}}$.  Then Protocol~{\REp} is
  $\zeta$-secure with respect to~{\IDp} for $\zeta=\frac{23}{2}N^{-1/2}$.
  Furthermore, in a noise-free implementation with honest devices,
  Protocol~{\REp} does not abort with probability greater than
  $(1+\pi^2N^{-1/2}/16)^{-1}$.
\end{theorem}
\begin{proof}
  As mentioned above, Protocols~{\REp} and~{\IDp} can be optimally
  distinguished when the eavesdropper supplies i.i.d.\ states, and so
  her device behaviour can be characterized by a single value, $I_N^*$, the
  value of $I_N(P_{X_iY_i|A_iB_i})$ on each round $i$.  The two
  protocols are identical up to Step~\ref{step6}, and so can be
  distinguished only if the protocol does not abort.  In the case of
  no abort, the distinguisher sees $A_f$, $B_f$, $X'_f$ and $Y_f$, and
  then has access to a system with input $C$ and output $Z$.  (The
  distinguisher also has data from previous rounds, but these are
  identically distributed for Protocols~{\REp} and~{\IDp}, and so can
  be ignored.)  Noting that the device behaviour of the ideal obeys
  $$P^{\text{\IDp}}_{X'YZ|abc}(x,y,z):=\frac{1}{2}\delta_{x,y}P^{\text{\REp}}_{Z|c}(z)\, ,$$
  we can relate the terms in~\eqref{eq:defn} to the device behaviours of the
  real and ideal as follows:
  \begin{eqnarray*}
    Q^{\text{\REp}}_{A_fB_fX'_fY_f}&=&\frac{1}{2N}P^{\text{\REp}}_{X'_fY_f|A_fB_f}\\
    Q^{\text{\REp}}_{Z|A_fB_fCX'_fY_f}&=&P^{\text{\REp}}_{Z|A_fB_fCX'_fY_f}\\
    Q^{\text{\IDp}}_{A_fB_fX'_fY_f}&=&\frac{1}{4N}\delta_{x,y}\\
    Q^{\text{\IDp}}_{Z|A_fB_fCX'_fY_f}&=&P^{\text{\REp}}_{Z|C}\, .
  \end{eqnarray*}
  For convenience, we drop the subscript $f$ in the following.

  We will consider two separate cases.  The first is $I_N^*\geq 1/2$.
  In this case, we can upper bound the probability of correctly
  distinguishing the protocols by assuming that they can be perfectly
  distinguished in the case that the protocol does not abort.  Using
  Lemma~\eqref{lem:2}, it follows that in this case the probability of
  correctly guessing which protocol is being used can be upper bounded
  by
  $$\frac{1}{2}(1+Q(\bar{\bot}))\leq\frac{1}{2}(1+4N^{-\frac{1}{2}})\, ,$$
  where we have substituted the value of $\alpha$ and used
  \begin{equation}\label{eq:bound}
    (1+4N^{-\frac{1}{2}}-N^{-\frac{3}{2}})^{-1}\leq 1
  \end{equation}
  for $N\geq 2$ to simplify the bound.

  Turning now to the case $I_N^*\leq 1/2$, the probability of
  correctly guessing which protocol is being followed is
  $\frac{1}{2}(1+Q(\bar{\bot})\Delta)$, where
  \begin{widetext}
    \begin{eqnarray*}
      \Delta&:=&\frac{1}{2}\sum_{\genfrac{}{}{0pt}{}{a,b,x,y}{(a,b)\in\cG_N}}\max_c\sum_z|Q^{\text{\REp}}_{ABX'Y}Q^{\text{\REp}}_{Z|ABCX'Y}-Q^{\text{\IDp}}_{ABX'Y}Q^{\text{\IDp}}_{Z|ABCX'Y}|\\
      &=&\frac{1}{4N}\sum_{\genfrac{}{}{0pt}{}{a,b,x,y}{(a,b)\in\cG_N}}\max_c\sum_z|P^{\text{\REp}}_{X'Y|ab}(x,y)P^{\text{\REp}}_{Z|abcxy}(z)-\frac{1}{2}\delta_{x,y}P^{\text{\REp}}_{Z|c}(z)|\\
      &=&\frac{1}{4N}\!\!\!\sum_{\genfrac{}{}{0pt}{}{a,b,x,y}{(a,b)\in\cG_N,x=y}}\!\!\!\max_c\sum_z|P^{\text{\REp}}_{X'Y|ab}(x,y)P^{\text{\REp}}_{Z|abcxy}(z)-\frac{1}{2}P^{\text{\REp}}_{Z|c}(z)|+\frac{1}{4N}\!\!\!\sum_{\genfrac{}{}{0pt}{}{a,b,x,y}{(a,b)\in\cG_N,x\neq
          y}}\!\!\!\max_c\sum_zP^{\text{\REp}}_{X'Y|ab}(x,y)P^{\text{\REp}}_{Z|abcxy}(z)\,
      .
    \end{eqnarray*}

    The second term is equal to
    $\frac{1}{4N}\sum_{\genfrac{}{}{0pt}{}{a,b,x,y}{(a,b)\in\cG_N,x\neq
        y}}P^{\text{\REp}}_{X'Y|ab}(x,y)=\frac{1}{2}P^{\text{\REp}}(X'\neq
    Y)$, and the first term is equal to
    \begin{eqnarray*}
      \frac{1}{4N}\sum_{\genfrac{}{}{0pt}{}{a,b,x,y}{(a,b)\in\cG_N,x=y}}\max_c\sum_z|P^{\text{\REp}}_{X'Y|ab}(x,y)P^{\text{\REp}}_{Z|abcxy}(z)-\frac{1}{2}P^{\text{\REp}}_{Z|c}(z)|\,
      .
    \end{eqnarray*}
    Then note that
    \begin{eqnarray*}
      \sum_z|P^{\text{\REp}}_{X'Y|ab}(x,x)P^{\text{\REp}}_{Z|abcxy}(z)-\frac{1}{2}P^{\text{\REp}}_{Z|c}(z)|&\leq&
      \sum_z|P^{\text{\REp}}_{X'Y|ab}(x,x)P^{\text{\REp}}_{Z|abcxy}(z)-P^{\text{\REp}}_{X'Y|ab}(x,x)P^{\text{\REp}}_{Z|c}(z)|\\
      &&+\sum_z|P^{\text{\REp}}_{X'Y|ab}(x,x)P^{\text{\REp}}_{Z|c}(z)-\frac{1}{2}P^{\text{\REp}}_{Z|c}(z)|\\
      &=&\sum_zP^{\text{\REp}}_{X'Y|ab}(x,x)|P^{\text{\REp}}_{Z|abcxx}(z)-P^{\text{\REp}}_{Z|c}(z)|+\sum_zP^{\text{\REp}}_{Z|c}(z)|P^{\text{\REp}}_{X'Y|ab}(x,x)-\frac{1}{2}|\\
      &\leq&2P^{\text{\REp}}_{X'Y|ab}(x,x)(\frac{2I_N^*}{1-I_N^*}+I_N^*)+|P^{\text{\REp}}_{X'Y|ab}(x,x)-\frac{1}{2}|\,
      .
    \end{eqnarray*}
    where we have used~\eqref{eq:together}.  In addition,
    \begin{eqnarray*}
      |P^{\text{\REp}}_{X'Y|ab}(x,x)-\frac{1}{2}|&\leq&|P^{\text{\REp}}_{X'Y|ab}(x,x)-P^{\text{\REp}}_{X'|ab}(x)|+|P^{\text{\REp}}_{X'|ab}(x)-\frac{1}{2}|\\
      &=&P^{\text{\REp}}_{X'|ab}(x)-P^{\text{\REp}}_{X'Y|ab}(x,x)+|P^{\text{\REp}}_{X'|ab}(x)-\frac{1}{2}|\, .
    \end{eqnarray*}

    Bringing everything together, we have
    \begin{align*}
      \Delta&\leq\frac{1}{4N}\sum_{\genfrac{}{}{0pt}{}{a,b,x}{(a,b)\in\cG_N}}\left(P^{\text{\REp}}_{X'Y|ab}(x,x)(\frac{4}{1-I_N^*}+2)I_N^*+P^{\text{\REp}}_{X'|ab}(x)-P^{\text{\REp}}_{X'Y|ab}(x,x)+|P^{\text{\REp}}_{X'|ab}(x)-\frac{1}{2}|\right)+\frac{1}{2}P^{\text{\REp}}(X'\neq
      Y)\\
      &\leq(\frac{2}{1-I_N^*}+1)I_N^*+\frac{1}{2}(1-P^{\text{\REp}}(X'=Y))+\frac{I_N^*}{2}+\frac{1}{2}P^{\text{\REp}}(X'\neq
      Y)=\left(\frac{2}{1-I_N^*}+\frac{3}{2}+\frac{1}{2N}\right)I_N^*\leq\frac{23}{4}I_N^*\,
      ,
    \end{align*}
  \end{widetext}
  where we used~\eqref{eq:distX},~\eqref{correctbound}, and the last
  bound relies on $I_N^*\leq 1/2$ and $N\geq 2$.  The distinguisher's
  probability of correctly guessing is thus
  \begin{eqnarray*}
    \frac{1}{2}(1+Q(\bar{\bot})\Delta)\leq
    \frac{1}{2}\left(1+\left(1+\frac{(1-\alpha)I_N^*}{2N\alpha}\right)^{-1}\frac{23}{4}I_N^*\right)\, .
  \end{eqnarray*}
  Maximizing over $0\leq I_N^*\leq 1/2$ gives a maximum of
  $\frac{1}{2}(1+\frac{23N\alpha}{2(4N\alpha+1-\alpha)})$ at
  $I_N^*=1/2$.  Substituting $\alpha=N^{-\frac{3}{2}}$ and
  using~\eqref{eq:bound}, we can upper bound this by
  $\frac{1}{2}(1+\frac{23}{2}N^{-\frac{1}{2}})$.  Since we have
  already established a tighter bound for $I_N^*\geq\frac{1}{2}$, this
  completes the first part of the claim.

  The probability of an abort in the case that Eve supplies honest
  devices (and there is no noise) can be calculated as in
  Lemma~\ref{lem:2}, except that in this situation, each round has
  $I_N=I_N^{\QM}<\pi^2/8N$ (cf.~\eqref{qcbi}).  The probability that
  the protocol does not abort is then
  \begin{equation}\label{honestabortprob}
    \left(1+\frac{(1-\alpha)I_N^{\QM}}{2N\alpha}\right)^{-1}>\left(1+\frac{\pi^2}{16N^2\alpha}\right)^{-1}
    \, ,
  \end{equation}
  from which the claim is recovered by substituting the value of
  $\alpha$.
\end{proof}

For sufficiently large $N$, we can hence make $\zeta$ as close to $0$
as we like, at the same time as making the probability of an abort in
the absence of Eve close to $0$.

Finally, since, by construction, it is harder to distinguish
Protocol~{\RE} from Protocol~{\ID} than it is to distinguish
Protocol~{\REp} from~{\IDp}, Protocol~{\RE} is also $\zeta$-secure
with respect to Protocol~{\ID} for the same $\zeta$, and an analogous
statement can be made about Protocol~{\REpl}.  Clearly, too, when Eve
is honest and noise is absent, Protocols~{\RE},~{\REp} and~{\REpl} all
have the same abort probability, in each case bounded
by~\eqref{honestabortprob}.

\section*{Attacks on modified protocols by a post-quantum
  eavesdropper}

Protocol~{\RE} relies on a probabilistic strategy in which Alice and
Bob sequentially either (with high probability) test a purported
entangled state generated by their devices or (with low probability)
generate a key bit from the state and immediately end the protocol.
We consider below two seemingly natural modifications of
Protocol~{\RE} and highlight some interesting attacks available to an
eavesdropper in such cases.  The first of our modified protocols can
be broken by a quantum eavesdropper and that the second can be broken
by an eavesdropper restricted only by signalling constraints.

\subsection*{Protocol~{\REt}}
This protocol is specified by positive integers $M$ and $N$.
\begin{enumerate}
\item On the $i^{\text{th}}$ round, Alice picks a pair of values
  $(A_i,B_i)$ at random from the set $\cG_N$, and announces $B_i$ to
  Bob.
\item Alice inputs $A_i$ into her device, and Bob $B_i$ into his, and
  they record their outcomes, the bits $X_i$ and $Y_i$
  respectively. (Alice ensures that her device doesn't learn $B_i$.)
  If $(A_i,B_i)=(0,2N-1)$, Alice sets $X'_i=1-X_i$, otherwise
  she sets $X'_i=X_i$.  The protocol returns to Step~1 unless
  $i=M$.
\item Alice randomly chooses an integer $1\leq f\leq M$ and
  announces it to Bob.
\item \label{step:adj} Alice and Bob publicly announce $X'_i$
  and $Y_i$ for all $i\neq f$.  If any of their announced values are
  unequal, they abort.
\item The bits $X'_f$ and $Y_f$ are taken to be the final
  shared bit.
\end{enumerate}
This protocol is similar in spirit to the original BHK
protocol~\cite{BHK}, and vulnerable to the same kind of attack in the
scenario where Alice and Bob have only one device each.

In this case, if Eve equips her devices with memory, she has a simple
attack.  She programs her devices to behave honestly until the final
($M^{\text{th}}$) round.  On this round, Alice's device outputs the
{\sc xor} of the previous outputs, i.e.\ $\bigoplus_{i=1}^{M}X_i$, and
Bob's device outputs a random bit.  This attack leads to a probability
of abort close to $\frac{1}{2}$, and otherwise enables Eve to
perfectly guess the final output bit.  Crucially, the success
probability of this strategy cannot be made small by adjusting $M$ and
$N$.

Define Protocol~{\REc} by altering  Step~\ref{step:adj} of Protocol~{\REt} to
circumvent this attack: 
\begin{enumerate}
\item[$\bar{\ref{step:adj}}$.] For all $i\neq f$, Alice chooses
  $L_i=0$ with probability $\beta$ and $L_i=1$ with probability
  $1-\beta$.  She announces this list to Bob.  For all the rounds in
  which $L_i=1$, Alice and Bob publicly announce their outcomes.  If
  any of their announced values are unequal, they abort.
\end{enumerate}
With this modification, making the final output the {\sc xor} of the
previous ones does not give Eve significant information, since Eve no
longer learns all but one of the outputs, $\{X_i\}$.  However, there
is another attack that a post-quantum non-signalling eavesdropper can use in this case, which
allows her to learn the final bit, again with a probability of success
that cannot be made small for any choice of $M$ and $N$.  This attack
exploits some subtle properties of non-local correlations and cannot
be performed by a quantum-limited eavesdropper.

The attack is based on a result in~\cite{vD} and involves non-local
boxes~\cite{Cirelson93,PR}.  These are bipartite systems where each
party has two choices of input and receives one of two outputs.  If we
denote the inputs $x\in\{0,1\}$ and $z\in\{0,1\}$ and the respective
outputs $\alpha\in\{0,1\}$ and $\gamma\in\{0,1\}$, then the non-local
box is a non-signalling device which outputs according to
$x.z=\alpha\oplus\gamma$.

The attack is as follows.  Eve constructs Alice's device such
that it contains both a set of maximally entangled quantum states
shared with Bob, and a set of non-local boxes shared with Eve (the
same number of each). For the first $M-\frac{1}{\beta}$ rounds of the
protocol, Alice's device generates its output by making quantum
measurements as in an honest implementation of the protocol.  However,
as well as supplying the measurement outcome to the output port of the
device (so that Alice sees it), the outcome is also used as input to
one of the non-local boxes, generating an output (call it $\alpha_i$).
(Bob's device behaves honestly in the first $M-\frac{1}{\beta}$
rounds, and outputs predetermined random bits in the remaining ones.)

In the last $\frac{1}{\beta}$ rounds, Alice's device instead always
outputs the {\sc xor} of all the previous non-local box outputs, i.e.\
$\bigoplus_i\alpha_i$.  (Although this may look suspicious, it does
not violate the stated security tests.  In any case it could
easily be masked by shared randomness between Alice's device and Eve.)
With reasonable probability, Eve will learn this bit (on each round of
the protocol, the chances that the output of that round is
communicated between Alice and Bob is $\beta$, so, of the last
$\frac{1}{\beta}$, on average 1 will be communicated).  For each bit
of the last $\frac{1}{\beta}$ that is communicated there is a probability
$1/2$ of being detected by Alice and Bob, so this strategy implies a
significant probability that Eve will be detected.  However, the
probability that this attack works without detection is independent of
$N$ and $M$ and at least $\frac{1}{2e}$.

If Eve learns $\bigoplus_i\alpha_i$, she can determine the key bit,
$x_f$.  To see this, notice the non-local box condition is
$x_i.z_i=\alpha_i\oplus\gamma_i$, where $z_i$ are the inputs and
$\gamma_i$ the outputs of Eve's half of the non-local box.  Eve should
input $0$ to all of her halves of the non-local boxes, except the
$f^{\text{th}}$ one in which she inputs $1$.  We
have 
$$x_f=\bigoplus_i(x_i.z_i)=\bigoplus_i(\alpha_i\oplus\gamma_i)=\bigoplus_i\alpha_i+\bigoplus_i\gamma_i.$$
Therefore, provided she has obtained the bit $\bigoplus_i\alpha_i$,
Eve can determine the final bit output by the protocol, $x_f$.

\section*{Attacking more noise-tolerant protocols}
In this section, we consider some extensions of the type of attack
considered in the previous section to two-device protocols that (if
secure) would be more efficient and tolerate more noise.

In all device-independent key distribution protocols, one needs, in
essence, to establish the presence of non-local correlations.  In
order to do so, the detection loophole must be closed.  In other
words, a malicious device should not be able to exploit detector
failures (cases where no outcome is observed) to give the false
illusion of non-locality in the non-failure cases.

Protocols based on chained Bell correlations with large $N$, are not
well-suited to this, since as $N$ increases, it becomes increasingly
difficult to close the detection loophole (the correlations can be
classically explained if the probability of detector failure is
$\frac{1}{N}$).  This drawback is not limited to the two-device case,
and alternative protocols tolerating modest levels of noise have been
introduced in the case where more devices are permitted~\cite{HR,MPA}.
We now consider the extension of these protocols to the two-device
case.  We do not give a proof that all such protocols are insecure,
but give an example that highlights interesting security issues that
can arise in the presence of non-signalling eavesdroppers.

We also mention some other work related to this question.
In~\cite{HRW}, the two-device case was considered for protocols based
on CHSH correlations.  There it was shown that privacy amplification
via hashing is not possible against an adversary limited only by the
impossibility of signalling between the parties.  However,
in~\cite{HRW}, signalling was permitted within the devices (so that
outputs could depend on later inputs\footnote{Although, as currently
  described, this is unphysical, it is natural to consider this for
  protocols in which each party makes all their inputs at the start,
  and then receives all of their outputs together.}).  For protocols
in which each party waits for an output before giving their next
input, the most natural signalling constraints are ones that allow
later outputs to depend on all previous inputs, but do not allow
outputs to depend on future inputs (we call these \emph{time-ordered
  non-signalling conditions}).  A situation that is close to this case
(but with subtle and potentially important differences) has been
recently studied in~\cite{AHT}.  There protocols based on CHSH
correlations were again considered, and it was shown that privacy
amplification via hashing is not possible for adversaries limited by
almost time-ordered non-signalling conditions.

Consider now a key distribution protocol with the following
structure\footnote{\label{ft:note}
  Although this structure is not fully general, most protocols to date
  are of this type.}:
\begin{enumerate}
\item Alice and Bob each make a random input $A_i$ and $B_i$ to their
  devices, ensuring they receive their outputs ($X_i$ and $Y_i$
  respectively) before making the next input (so that time-ordered
  non-signalling conditions must be obeyed).  They repeat this $M$
  times.
\item Either Alice, or Bob (or both) publicly announces their
  measurement choices, and one party checks that they had a sufficient
  number of the relevant input combinations, and otherwise aborts.
  Certain rounds may be discarded according to some public protocol.
\item \label{step:param_est} For each of the remaining bits, Alice
  independently announces it to Bob with probability $\mu$ (which is
  such that $M\mu$ is large).  Bob uses this to compute some test
  function.  If this has the wrong output, Bob aborts. (For example,
  Bob might compute the CHSH value of the announced data, and abort if
  it is below $2.5$.)
  (This step is often called \emph{parameter estimation}.)
\item Alice and Bob perform error correction using public
  communication via any protocol in which the function Alice applies
  to her string becomes known to Eve\footnote{\label{ft:f}Typically
    because it is communicated over the public channel.}.
\item Alice and Bob publicly perform privacy amplification.  The
  function Alice applies to her string becomes known to Eve$^{\ref{ft:f}}$.
\end{enumerate}

The key to Eve's attack is Step~\ref{step:param_est}.  She is going to
attack so as to try to gain one bit of the final output string.  Eve
will also use a ``joint function box'', which has the following
bipartite behaviour.  Alice inputs a string $X_1\ldots X_M$ and
obtains a single bit $S$, Eve inputs $C$ which corresponds to a choice
of one-bit function (see later), and obtains a single bit $Z$.  The
behaviour is such that $Z=S\oplus F_C(X_1\ldots X_M)$.  It is easy to
see that this can be non-signalling if $S$ is a uniform random bit.

\begin{table}
\begin{tabular}{ccc|cc|cc|cc|}
$P_{SZ|XC}$&&$C$&\multicolumn{2}{|c|}{0}&\multicolumn{2}{|c|}{1}&\multicolumn{2}{|c|}{2}\\
&&$Z$&0&1&0&1&0&1\\
$X$&$S$&&&&&&&\\
\hline
\multirow{2}{*}{00\ldots 00}&0&&1/2&0&0&1/2&\ldots&\\
&1&&0&1/2&1/2&0&&\\
\hline
\multirow{2}{*}{00\ldots 01}&0&&0&1/2&1/2&0&&\\
&1&&1/2&0&0&1/2&&\\
\hline
&\vdots&&&&\vdots&&
\end{tabular}
\caption{{\bf Behaviour of the ``joint function box''.} Each $2\times
  2$ block takes one of the two forms shown, depending on whether
  $F_C(X_1\ldots X_M)=0$ or $F_C(X_1\ldots X_M)=1$.  In this
  notation, the non-signalling conditions are that the sum of the
  elements in each row of each $2\times 2$ block are equal to those of
  the blocks to the left and right, and likewise, the sum of the
  elements in each column are equal to those above and below.  In the
  above case, all of these values are $1/2$.}
\end{table}

It follows that Eve can learn any Boolean function of $X_1,\ldots,
X_M$ if she receives just one bit, $S$.  The value of $C$ depends on
the function Eve wants to learn, the value of $S$ she hears and the
information reconciliation and privacy amplification functions she
overhears.  There is a choice of $C$ for each combination of these
values.  Thus, for protocols of the above form (importantly, where Eve
learns the entire function Alice and Bob use for post-processing), she
needs to receive only one bit from either of her devices to learn one
bit about the final output key (after privacy amplification).

In order to try to learn this bit, Eve can exploit the parameter
estimation step.  She programs Alice's device to behave honestly for
the first $M-1/\mu$ rounds (note that we have not specified which
correlations are used; this attack does not depend on these (up to a
constant factor in the abort probability) and even works if the honest
states are perfect non-local boxes).  Her device then inputs the bits
generated in these rounds into the ``joint function box'', producing
output $S$.  This bit is then given as the outputs $X_i$ for
$M-1/\mu\leq i\leq M$ (this could be masked by {\sc xor}ing with some
pre-shared randomness between Alice's device and Eve).  Provided at
least one of the last $1/\mu$ bits is revealed in the parameter
estimation without causing abort (this occurs with finite probability
that cannot be made arbitrarily small by judiciously choosing $M$ and
$\mu$), Eve can discover any desired bit of the final output string.

There are a couple of important points to note about the above attack.
Firstly, we assumed a specific protocol structure.  In particular,
altering the way parameter estimation is done could potentially
improve security (some altered protocols are discussed
in~\cite{bckone}).  Secondly, the attack relies on a specialized
non-local strategy that cannot be implemented by an eavesdropper
limited by quantum theory.  Proving security of a protocol of this
type (in particular, with two devices) that is secure against a
quantum-restricted Eve remains an open question.

\section*{Conclusions}
We have presented a protocol for distribution of a one-bit key, and
have proven it secure in a universally composable way against an
arbitrarily powerful adversary who can create all the supposedly
quantum devices, provided that the devices are not reused in any
future protocol.  The protocol only requires two devices, whereas the
secure protocols previously considered required many independent
devices.  This represents a theoretical advance, and also potentially
represents another step towards practical unconditionally secure
device-independent key distribution protocols.

That said, several significant and intriguing theoretical and
practical issues remain.  First, the simplest version of our protocol
only outputs a single bit, requiring a large number of entangled qubit
pairs in order to do so.  The protocol can be generalised to produce
an arbitrary length key string, but again, highly inefficiently.  It
would be very interesting to know whether significantly more efficient
two-device secure protocols can be found, and to obtain bounds on what
is achievable.

Secondly, for maximum flexibility and more efficient use of resources,
one would like to be able to repeat the protocol to generate further
secure key bits.  However, if devices are reused, this renders the
protocol vulnerable to the same device-memory-based
attacks~\cite{bckone} that apply to BHK and other device-independent
protocols.  While it is clear that device-reusing protocols cannot be
universally composable, the general scope of such attacks and the
possibilities of countering them either by refined protocols
(see~\cite{bckone} for ideas in this direction, some of which have
been later developed in~\cite{MS}) or by evidently reliable
technological assumptions have not yet been fully explored.  It would
be very interesting to resolve these questions in the present context.

Thirdly, tolerance to noise is a significant practical issue for our
protocol.  As given, it aborts if there is one set of measurements
that give unequal outcomes.  The protocol parameters are tuned such
that this is very unlikely if the devices operate perfectly.  However,
with more realistic, noisy devices, using present technology, the
protocol will have a very high abort rate.  Although the protocol
could be adapted to tolerate small amounts of noise, it is far from
being practical in this respect.

Fourthly, our scheme requires an authenticated (although public)
classical channel, and, a common way to implement this in an
information-theoretically secure way using an insecure classical
channel, is by using a pre-shared key.  This reinforces the points
already made: it would be desirable to have more efficient two-device
protocols that allow for some consumption of key for classical
authentication and nonetheless provide quantum key expansion at
practically useful rates in realistically noisy environments.

In summary, while we have presented a protocol showing that
device-independent quantum key distribution is in principle possible
using two devices, a number of theoretically interesting and
practically important questions remain open.\bigskip

{\bf Remark:} In some concurrent work an alternative technique for
proving security of device-independent QKD with two devices has been
suggested~\cite{RUV}.  Furthermore, since the first version of our
paper, an additional article has appeared~\cite{VV2} reporting an
efficient and noise-tolerant scheme.  We note that these works differ
from ours in that they consider quantum-limited eavesdroppers and do
not apply to the case of eavesdroppers limited only by signalling
constraints.

\begin{acknowledgments}
  We thank Graeme Mitchison for many useful conversations.  JB was
  supported by the EPSRC, and the CHIST-ERA DIQIP project.  RC
  acknowledges support from the Swiss National Science Foundation
  (grants PP00P2-128455 and 20CH21-138799) and the National Centre of
  Competence in Research `Quantum Science and Technology'.  AK was
  partially supported by a Leverhulme Research Fellowship, a grant
  from the John Templeton Foundation, and the EU Quantum Computer
  Science project (contract 255961).  This research is supported in
  part by Perimeter Institute for Theoretical Physics.  Research at
  Perimeter Institute is supported by the Government of Canada through
  Industry Canada and by the Province of Ontario through the Ministry
  of Research and Innovation.
\end{acknowledgments}


\end{document}